\documentclass[runningheads]{llncs}
\usepackage[T1]{fontenc}
\usepackage{graphicx}
\usepackage{amssymb}
\usepackage{amsmath}
\usepackage{mathtools}
\usepackage{xcolor}
\usepackage{txfonts}
\usepackage{enumerate}
\usepackage{bbm}
\usepackage{url}
\usepackage{xspace}
\usepackage[normalem]{ulem}
\usepackage{tikz}
\usetikzlibrary{automata,positioning}
\usepackage{multirow}
\usepackage{array}
\usepackage{booktabs}
\usepackage{wrapfig}

\usepackage{hyperref} 

\usepackage{breakcites}

\newcommand{\dist}{\mathcal{D}}
\newcommand{\tup}[1]{\left(#1\right)}

\newcommand{\set}[1]{\left\lbrace #1\right\rbrace}

\newcommand{\RR}{\mathbb{R}}

\renewcommand{\SS}{{\mathcal{X}}}
\newcommand{\IS}{\mathcal{U}}

\newcommand{\Sys}{\mathcal{S}}
\newcommand{\cont}{\pi}

\renewcommand{\path}{\rho}

\newcommand{\AP}{P}

\newcommand{\DS}{\mathcal{W}}
\newcommand{\Dyn}{f}
\newcommand{\Dynx}{f^\times}
\newcommand{\Sysx}{\Sys^\times}
\newcommand{\SSx}{\mathcal{X}^\times}
\newcommand{\ISx}{\mathcal{U}^\times}

\newcommand{\Sysrw}{\Sys^{\mathrm{rw}}}

\newcommand{\Aut}{\mathcal{A}}
\newcommand{\Q}{Q}
\newcommand{\QN}{\Q_n}
\newcommand{\QD}{\Q_d}
\newcommand{\qinit}{q_{\mathsf{init}}}
\newcommand{\acc}{\mathcal{F}}
\newcommand{\tran}{\Delta}

\newcommand{\Qreject}{\Q_{\text{reject}}}

\newcommand{\Spec}{\varphi}

\newcommand{\lsm}{LDBSM\xspace}
\newcommand{\lsms}{LDBSMs\xspace}
\newcommand{\Cert}{V}

\newcommand{\safe}{\mathrm{safe}}
\newcommand{\CertS}{V^{\safe}}
\newcommand{\live}{\mathrm{live}}
\newcommand{\CertB}{V^{\live}}
\newcommand{\safetycond}{\mathit{SafetyCond}(x,q,a,q')}

\newcommand{\Init}{\text{Init}}
\newcommand{\Initx}{\text{Init}^\times}

\newcommand{\var}[1]{\underline{#1}}

\newcommand{\F}{\textsc{F}}
\newcommand{\G}{\textsc{G}}
\newcommand{\U}{\textsc{U}}

\newcommand{\SI}{\textit{SI}}

\newcommand{\Buchi}{\textit{B\"uchi}}

\usepackage{comment}
\usepackage{paralist}

\begin{document}
%
\title{Supermartingale Certificates for \\ Quantitative Omega-regular Verification and Control}
%
%
\author{Thomas A.\ Henzinger\inst{1} \and
Kaushik Mallik\inst{2} \and\\
Pouya Sadeghi\inst{3} \and
\DJ or\dj e \v{Z}ikeli\'c\inst{3}}
%
%
\institute{Institute of Science and Technology Austria (ISTA), Austria\\ \email{tah@ist.ac.at} \and 
IMDEA Software Institute, Spain\\ \email{kaushik.mallik@imdea.org} \and
Singapore Management University, Singapore\\ \email{pouyas@smu.edu.sg}, \email{dzikelic@smu.edu.sg} }
%
\maketitle              
\begin{abstract}
We present the first supermartingale certificate for quantitative $\omega$-regular properties of discrete-time infinite-state stochastic systems. Our certificate is defined on the product of the stochastic system and a limit-deterministic B\"uchi automaton that specifies the property of interest; hence we call it a limit-deterministic B\"uchi supermartingale (LDBSM). 
Previously known supermartingale certificates applied only to quantitative reachability, safety, or reach-avoid properties, and to qualitative (i.e., probability~1) $\omega$-regular properties.
We also present fully automated algorithms for the template-based synthesis of LDBSMs, for the case when the stochastic system dynamics and the controller can be represented in terms of polynomial inequalities. Our experiments demonstrate the ability of our method to solve verification and control tasks for stochastic systems that were beyond the reach of previous supermartingale-based approaches.

\keywords{Supermartingales  \and Probabilistic verification \and Stochastic control}
\end{abstract}

\section{Introduction}\label{sec:intro}

Stochastic (or probabilistic) systems provide a framework for modeling and quantifying uncertainties in computational models. They are ubiquitous in many areas of computer science, including randomized algorithms~\cite{segala2000verification}, security and privacy protocols~\cite{BartheGGHS16}, stochastic networks~\cite{FosterKMR016}, control theory~\cite{aastrom2012introduction}, and artificial intelligence~\cite{Ghahramani15}. Many of these systems are safety critical in nature; hence formal methods for guaranteeing their correctness have become an important topic of research.

While the correctness of stochastic systems can be accomplished using a variety of complementary approaches, such as abstraction-based~\cite{esmaeil2013adaptive,zamani2014symbolic,dutreix2022abstraction,majumdar2024symbolic}, symbolic integration-based~\cite{GehrMV16,SankaranarayananCG13,BeutnerOZ22}, or logical calculi-based methods~\cite{MM05,KKMO16:wp-expected-runtime,OKKM16:recursive-prob-wp-calculus}, there is a growing interest in martingale-based approaches, which do not only guarantee correctness but also provide \emph{supermartingale certificates} as locally checkable correctness proofs for improving the trustworthiness of stochastic system design. A {\em supermartingale certificate} is a mathematical witness that the desired specification is satisfied. Their name is due to their usage of supermartingale processes from probability theory~\cite{Williams91} towards proving properties of stochastic systems.

Recent years have seen significant theoretical and algorithmic advances in utilizing supermartingale certificates for reasoning about stochastic systems with a wide range of specifications, such as termination and reachability~\cite{SriramCAV,CFNH16:prob-termination,CFG16,mciver2017new,abate2021learning,chatterjee2023lexicographic,majumdar2025sound}, safety~\cite{CS14:expectation-invariants,ChatterjeeNZ17,takisaka2021ranking,WangS0CG21,chatterjee2022sound,AbateEGPR23,WangYFLO24}, B\"uchi and co-B\"uchi (aka, stability)~\cite{ChakarovVS16}, and cost analysis~\cite{ngo2018bounded,WFGCQS19,chatterjee2024quantitative}. These certificates have also been used for the synthesis and verification of controllers with respect to reachability~\cite{florchinger1995lyapunov,lechner2022stability}, safety~\cite{prajna2004stochastic,steinhardt2012finite,mathiesen2022safety,ZhiWLOZ24}, reach-avoid~\cite{vzikelic2023learning,xue2023reach}, and stability~\cite{ansaripour2023learning} objectives, and extensions to continuous-time stochastic systems were also considered~\cite{NeustroevGL25}. However, existing supermartingale certificates are limited to atomic specifications or to small fragments of logical specification languages. Designing a new, bespoke supermartingale certificate for each complex automaton-based or composite logical specification is not a feasible approach. Ideally, we would like to have a systematic method for designing supermartingale certificates for an entire, rich {\em class of specifications}, which captures all different properties that one might care about. The recent work of Abate et al.~\cite{abate2024stochastic} has considered this problem and proposed a supermartingale certificate for the general class of $\omega$-regular specifications. However, their certificate is restricted to properties that hold with probability~1, so-called {\em qualitative} specifications. To the best of our knowledge, no supermartingale certificates for {\em quantitative $\omega$-regular properties} have been proposed. That is, existing supermartingale certificates are insufficient for reasoning about whether an $\omega$-regular specification is satisfied with a probability $p \in [0,1]$, which would allow the fine-grained probabilistic analysis that is often needed in practice.

In this work, we present the first supermartingale certificate for reasoning about the full class of {\em quantitative $\omega$-regular specifications} for discrete-time infinite-state stochastic systems. In order to design our certificate, we use the fact that each $\omega$-regular language can be represented by a limit-deterministic B\"uchi automaton (LDBA)~\cite{CourcoubetisY95}. This allows us to reduce the problem of proving that an $\omega$-regular specification is satisfied with probability at least $p \in [0,1]$ to the problem of proving that some set of states in the product of the stochastic system and the LDBA is reached infinitely many times with probability at least~$p$. Hence, we call our certificiate a {\em limit-deterministic B\"uchi supermartingale (\lsm)}. The design of \lsms builds upon and significantly generalizes several existing supermartingale certificates for both quantitative and qualitative reasoning about stochastic systems.

We supplement our theoretical developments with {\em fully automated algorithms for the verification and control} of infinite-state stochastic systems using \lsms. Our algorithms apply to the setting in which the system dynamics and the controller can be represented in terms of a system of polynomial inequalities over real-valued variables. We show how the verification problem (i.e., the synthesis of a supermartingale certificate) and the control problem (i.e., the synthesis of a controller together with a supermartingale certificate) can be reduced to solving a system of polynomial inequalities, for which an off-the-shelf SMT solver may be used. Not only does this give rise to the first \emph{certificate-based} algorithm, but also the \emph{first} algorithm for the verification and control of general \emph{quantitative $\omega$-regular specifications} for polynomial dynamical systems with a general form of stochastic uncertainties. We implemented a prototype of our method and experimentally evaluated it on a number of random walk examples with quantitative $\omega$-regular specifications. Our results show that our method can reason about a range of quantitative $\omega$-regular properties, including those that were beyond the reach of prior works. 

Finally, while our focus is on the verification and control of infinite-state stochastic systems, our \lsm certificates are applicable to general Borel-measurable stochastic dynamic systems, and hence they naturally extend to probabilistic programs with real- and integer-valued variables. Moreover, our automated synthesis algorithm also extends to polynomial probabilistic programs, similarly to existing works on supermartingale-based verification of probabilistic programs~\cite{CFG16,ChatterjeeNZ17,chatterjee2022sound}.

In short, our contributions are as follows:
\begin{compactenum}
	\item {\bf Theory: Supermartingale certificates.} We present limit-deterministic B\"uchi supermartingales (\lsms), the first supermartingale certificate for arbitrary quantitative $\omega$-regular specifications of discrete-time infinite-state stochastic systems.
	\item {\bf Automation: Template-based certificate synthesis.} We present the first fully automated algorithms for the verification and control of polynomial stochastic dynamical systems with respect to arbitrary quantitaitve $\omega$-regular specifications.
\end{compactenum}


\section{Preliminaries}\label{sec:prelims}

We assume that the reader is familiar with basic notions of probability theory such as probability space, expected value, etc, which can also be found in standard textbooks on probability theory~\cite{Williams91}. For a set $S$, we use $\dist(S)$ to denote the set of all probability distributions over $S$, and use $2^S$ to denote the set of all subsets of $S$. 

\subsection{Stochastic Dynamical Systems and the Problem Statement}\label{sec:problem}

\noindent{\bf Stochastic dynamical systems.} A discrete-time \emph{stochastic dynamical system} (SDS) is a tuple $\Sys = \tup{\SS,\IS,\DS,d,\Dyn,\Init}$, where $\SS \subseteq \RR^n$ is the state space, $\IS \subseteq \RR^m$ is the control input space, $\DS \subseteq \RR^p$ is the stochastic disturbance space, $d \in \dist(\DS)$ is the stochastic disturbance distribution, $\Dyn\colon \SS\times\IS\times\DS\to \SS$ is the {\em system dynamics} function, and $\Init \subseteq \SS$ is the set of initial system states. 
We refer to the elements of $\SS$ as {\em states}, the elements of $\IS$ as {\em control inputs}, and the elements of $\DS$ as {\em stochastic disturbances}. 
A {\em path} in $\Sys$ is an infinite sequence of states $x_0,x_1,\dots$ in $\SS$ where $x_0 \in \Init$ and for every $t\geq 0$, there exist a control input $u_t\in \IS$ and a stochastic disturbance $w_t\in \DS$ such that $x_{t+1} = f(x_t,u_t,w_t)$.

A \emph{controller} in $\Sys$ is a function $\cont: \SS\to \IS$, mapping every state to a control input. To avoid clutter, we define controllers as being memoryless, i.e., only depending on the current state and without using any information about the past. Our setting organically extends to {\em finite memory} controllers, which are able to remember bounded amount of information from the past. We will use one such finite memory controller in Section~\ref{sec:automata}, which will use its memory to track the current state of the specification automaton.

Our model of SDS subsumes the cases of SDS without control inputs (using $|\IS|=1$, making the control input ineffective) and without disturbances (using $|\DS|=1$, making the disturbance ineffective and giving rise to deterministic dynamical systems).

\smallskip\noindent{\bf Semantics of SDS.} We assume that the sets $\SS, \Init \subseteq \RR^n$, $\IS \subseteq \RR^m$ and $\DS \subseteq \RR^p$ as well as the functions $\Dyn\colon \SS\times\IS\times\DS\to \SS$ and $\pi:\SS \rightarrow \IS$ are all Borel-measurable. These assumptions are necessary for the semantics of SDS to be mathematically well defined.

Under these assumptions, for every initial state $x_0 \in \Init$, the SDS $\Sys$ and the controller $\cont$ together define a continuous-state, discrete-time Markov decision process that takes values in the set of states $\SS$ and whose trajectories correspond to paths in $\Sys$~\cite{Puterman94}. Initially, the process starts in $x_0$. Then, at every time step $t \in \mathbb{N}_0$, the next state of the process is defined by the equation $x_{t+1} = f(x_t,\pi(x)_t, w_t)$, where $w_t \sim d$ is a stochastic disturbance sampled from the stochastic disturbance distribution $d$, independently from the previous samples.

This process gives rise to a probability space over the set of all paths in $\Sys$~\cite{Puterman94}. We use $\mathbb{P}^{\Sys,\pi}_{x_0}$ and $\mathbb{E}^{\Sys,\pi}_{x_0}$ to denote the probability measure and the expectation operators in this probability space, respectively.

\smallskip\noindent{\bf Specifications.} A {\em specification} $\Spec$ in an SDS $\Sys$ is a set of paths in $\Sys$.
 We write $\mathbb{P}^{\Sys,\pi}_{x_0}[\Spec]$ to denote the probability that a path in $\Sys$ randomly sampled from the underlying probability space (over all paths)  satisfies the specification $\Spec$. We will consider $\omega$-regular specifications, which constitute a broad class of properties subsuming linear temporal logic (LTL) and computation tree logics~\cite{baier2008principles}. 
Every $\omega$-regular specification can be represented using limit deterministic B\"uchi automata, which will be defined in Section~\ref{sec:automata}.
We will occasionally use the standard LTL notation~\cite{baier2008principles} for convenience; in particular, $\G\, a$ will stand for ``always $a$,'' $\F\, a$ will stand for ``eventually $a$,'' and $a\U\, b$ will stand for ``$a$ until $b$.''

\smallskip\noindent{\bf Problem statement.} We now formally define the formal verification and control problems that we consider in this work. Consider an SDS $\Sys$, an $\omega$-regular specification $\varphi$ in $\Sys$, and a probability threshold $p \in [0,1]$:
\begin{compactenum}
	\item {\em Verification problem.} Given a controller $\pi$, prove that $\mathbb{P}^{\Sys,\pi}_{x_0}[\varphi] \geq p$ for all $x_0 \in \Init$.
	\item {\em Control problem.} Synthesize a controller $\pi$ such that $\mathbb{P}^{\Sys,\pi}_{x_0}[\varphi] \geq p$ for all $x_0 \in \Init$.
\end{compactenum}

\begin{example}[Running example: 1D random walk]\label{ex:random walk:1}
	For our running example, we will use a simple 1D random walk denoted as $\Sysrw$, whose state space is the real line $\mathbb{R}$, input space is a singleton set $\{\bot\}$ with a dummy input $\bot$, i.e., the system is uncontrolled, stochastic disturbance space is $ \DS=[-2,1]$ with a continuous uniform distribution over $\DS$, system dynamics is given as
	\begin{align}
		f(x,\bot,w) = 
			\begin{cases}
				x	&	\text{if } x>100,\\
				x + w	&	\text{otherwise,}
			\end{cases}
	\end{align}		
	 and the set of initial states is $[2,3]$.
	 Although this simple SDS does not contain control inputs in its dynamics, it will be sufficient for us to illustrate our technical contributions; a variation of $\Sysrw$ with control inputs will be considered in Section~\ref{sec:experiments}.
	Essentially, at every step, the random walk $\Sysrw$ has a higher probability of moving towards the left than towards the right. 
	If however, it ever crosses $100$, it becomes stationary.
	Consider the specification that is the set of all paths that visit the negative half of the real line infinitely many times, formalized as $\{ x_0x_1\ldots \mid \forall i\geq 0\;.\;\exists j>i\;.\; x_j\leq 0\}$. 
	Using the standard notation of linear temporal logic, we will write this specification as $\G\F\, (x\leq 0)$.
	For this example, we will refer to the set $\{x\mid x\leq 0\}$ as the \emph{target}.
\end{example}

\subsection{Limit-deterministic B\"uchi Automata}\label{sec:automata}

To specify properties of a given SDS, we label its state space with a finite set of {\em atomic propositions} $P$. 
A {\em labeling function} $L: \SS \rightarrow 2^P$ maps each state $x \in \SS$ to the set of atomic propositions that are true in $x$, and we assume that each atomic proposition $p\in P$ has an associated Borel-measurable arithmetic expression $\textrm{exp}_p: \SS \rightarrow \RR$  such that for every $x\in \SS$, $p\in L(x)$ iff $ \textrm{exp}_p(x)\geq 0$. Hence, a path $(x_i)_{i=0}^\infty$ in the SDS gives rise to an {\em infinite word} $(L(x_i))_{i=0}^\infty$ over the alphabet $2^P$.
 
We use the classical result that for every $\omega$-regular specification $\varphi$ defined over a finite set of atomic predicates $P$, there exists a limit-deterministic B\"uchi automaton with alphabet $\Sigma = 2^P$ which accepts the same set of infinite words over $2^P$ as $\varphi$~\cite{CourcoubetisY95}.

\smallskip\noindent{\bf Nondeterministic B\"uchi automata.} A {\em nondeterministic B\"uchi automaton (NBA)} is a tuple $\Aut = \tup{\Q,\qinit,\Sigma,\tran,\acc}$  where $\Q$ is a finite set of states, $\qinit\in\Q$ is an initial state, $\Sigma$ is a finite alphabet that includes the empty word symbol $\epsilon$,
$\tran\colon \Q\times \Sigma \to 2^\Q$ is a nondeterministic transition function, and $\acc\subseteq 2^\Q$ is a set of accepting states. An infinite word $\sigma_0,\sigma_1,\dots$ of letters in the alphabet $\Sigma$ is {\em accepted} by $\Aut$, if it gives rise to at least one B\"uchi accepting run in $\Aut$, i.e.,~if there exists a run $q_0,q_1,\dots$ of states in $\Q$ such that $q_0 = \qinit$, for each $i \in \mathbb{N}_0$, $q_{i+1} \in \Delta(q_i,\sigma_i)$, and for every $i\in \mathbb{N}_0$ there exists a $j>i$ such that $q_j\in \acc$ (infinitely many visits to $\acc$).

\smallskip\noindent{\bf Limit-deterministic B\"uchi automata.} {\em Limit-deterministic B\"uchi automata (LDBA)} are particular types of NBAs with a restricted form of nondeterminism. In particular, an NBA $\Aut = \tup{\Q,\qinit,\Sigma,\tran,\acc}$ is said to be an LDBA, if there exists a partition $\set{\QN,\QD}$ of the set of states $\Q$ (i.e., $\Q=\QN\cup \QD$ and $\QN\cap \QD=\emptyset$) such that for every $q\in\QD$ and for every $\sigma \in \Sigma$, (i)~the transitions from $q$ are deterministic, i.e., $|\tran(q,\sigma)|=1$, and (ii)~there is no transition from $q$ going outside of $\QD$, i.e., $\tran(q,\sigma)\subseteq \QD$. We call $\QN$ the {\em nondeterministic part} of the LDBA and $\QD$ the {\em deterministic part} of the LDBA.

As mentioned above, for every $\omega$-regular specification $\varphi$ defined over a finite set of atomic predicates $P$, there exists a canonical LDBA $\Aut$ with the alphabet $\Sigma = 2^P$ accepting the same set of infinite words as $\varphi$. Hence, a path $(x_i)_{i=0}^\infty$ in $\Sys$ satisfies $\omega$-regular specification $\varphi$ if and only if the infinite word  $(L(x_i))_{i=0}^\infty$ is accepted by the LDBA $\Aut$. Our verification and control problems thus reduce to ensuring that a random run in $\Sys$ induces an infinite word over $2^P$ accepted by the LDBA $\Aut$ with probability at least~$p$.

\smallskip\noindent{\bf Product SDS.} In order to reason about paths and infinite words that they induce in the LDBA, we consider a (synchronous) product of the SDS and the LDBA. A {\em product SDS} of the SDS $\Sys = \tup{\SS,\IS,\DS,d,\Dyn,\Init}$ and the LDBA $\Aut = \tup{\Q,\qinit,2^\AP,\tran,\acc}$ is a tuple $\Sysx = (\SSx,\ISx,\DS,d,\Dynx,\Initx)$, where
\begin{compactitem}
	\item $\SSx = \SS \times \Q$ is the state space of the product SDS,
	\item $\ISx = \IS \times \Q$ is the control input space of the product SDS,
	\item $\DS \subseteq \RR^p$ is the stochastic disturbance space of the product SDS,
	\item $d \in \dist(\DS)$ is the stochastic disturbance distribution of the product SDS,
	\item $\Dynx\colon \SSx\times\ISx\times\DS\to \SSx$ is the system dynamics function, and
	\item $\Initx = \Init \times \{\qinit\}$ is the set of initial states of the product SDS.
\end{compactitem}
Similarly as before, a {\em controller} in $\Sysx$ is a function $\pi^\times: \SSx \rightarrow \ISx$ mapping every state to a control input. However, we require that for every $(x,q) \in \SSx$, the corresponding control input $\pi^\times(x,q) = (u,q') \in \IS \times \Q$ satisfies $q' \in \Delta(q,L(x))$. Hence, the product SDS indeed models a synchronous product of the SDS and the LDBA, where the controller $\pi^\times$ resolves nondeterminism both in the SDS and in the LDBA. To distinghuish between the two controller components, we write $\pi^\times(x,q) = (\pi^\Sys(x,q),\pi^\Aut(x,q)))$. The semantics of the product SDS are defined analogously to Section~\ref{sec:problem}, with
\[ (x_{t+1},q_{t+1}) = \Big( f(x_t,\pi^\Sys(x_t,q_t),w_t),\pi^\Aut(x_t,q_t) \Big) \]
for every time step $t\in\mathbb{N}_0$ and $w_t$ sampled according to $d$ and independent from $w_i$ for every $i \leq t-1$. For each initial state $(x_0,\qinit) \in \Initx$, we denote by $\mathbb{P}^{\Sysx,\pi^\times}_{(x_0,\qinit)}$ and $\mathbb{E}^{\Sysx,\pi^\times}_{(x_0,\qinit)}$ the probability measure and the expectation operator in the probability space of all paths in $\Sysx$ with controller $\pi^\times$, respectively.

Finally, the following proposition will allow us to reduce our verification and control problems to the problems of synthesizing a controller for the product SDS. Denote by $\Buchi(B)$ the set of all runs in $\Sysx$ that visit states in $B \subseteq \SSx$ infinitely many times. 

\begin{proposition}\label{thm:product}
	Consider an SDS $\Sys$ and a controller $\pi$ in $\Sys$. Let $\varphi$ be an $\omega$-regular specification over a finite set of atomic propositions $P$, $\Aut = \tup{\Q,\qinit,2^P,\tran,\acc}$ be an LDBA for the specification $\varphi$, and $\Sysx$ be the product MDP of $\Sys$ and $\Aut$. Then, for every controller $\pi^\times = (\pi^\Sys,\pi^\Aut)$ in $\Sysx$ with $\pi^\Sys = \pi$ and for every initial state $(x_0,\qinit) \in \Initx$, we have
	\[ \mathbb{P}^{\Sys,\pi}_{x_0}\Big[ \varphi \Big] \geq \mathbb{P}^{\Sysx,\pi^\times}_{(x_0,\qinit)}\Big[ \Buchi(\mathbb{R}^n \times \acc) \Big]. \]
\end{proposition}

\begin{proof}
	Consider a controller $\pi^\times = (\pi^\Sys,\pi^\Aut)$ in $\Sysx$ with $\pi^\Sys = \pi$. Let $\Pi$ be the set of all paths in $\Sys$, and $\Pi^\times$ be the set of all paths in $\Sysx$. Define the map $g: \Pi \rightarrow \Pi^\times$ via
	\[ g(x_0,x_1,x_2,\dots) = ((x_0,q_0), (x_1,\pi^\Aut(x_0,q_0)), (x_2,\pi^\Aut(x_1,q_1)), \dots ). \]
	Now, observe that
	\begin{equation*}
	\begin{split}
		\mathbb{P}^{\Sys,\pi}_{x_0}\Big[ \varphi \Big] &= \mathbb{P}^{\Sys,\pi}_{x_0}\Big[ (x_i)_{i=0}^\infty \in \Pi \mid (x_i)_{i=0}^\infty \models \varphi \Big] \\
		&= \mathbb{P}^{\Sys,\pi}_{x_0}\Big[ (x_i)_{i=0}^\infty \in \Pi \mid (L(x_i))_{i=0}^\infty \text{ infinite word accepted by } \Aut\Big] \\
		&\geq \mathbb{P}^{\Sys,\pi}_{x_0}\Big[ (x_i)_{i=0}^\infty \in \Pi \mid (g(x_i))_{i=0}^\infty \text{ visits infinitely many times } \mathbb{R}^n \times \acc\Big] \\
		&=\mathbb{P}^{\Sysx,\pi^\times}_{(x_0,\qinit)}\Big[ \Buchi(\mathbb{R}^n \times \acc) \Big],
	\end{split}
	\end{equation*}
	where the inequality holds due to the first set of paths in $\Pi$ containing the second set of paths in $\Pi$. This proves the proposition claim. \hfill\qed
\end{proof}

The inequality in Proposition~\ref{thm:product} is strict  whenever $\pi^\Aut$ chooses a sub-optimal resolution of non-determinisms in the LDBA $\Aut$. Consider, e.g., the property $\F\G\, a$ (eventually always $a$), whose LDBA representation has three states $q_0, q_1, q_2$. The initial state is $q_0$, and upon seeing an $a$ the controller $\pi^{\Aut}$ needs to non-deterministically decide whether to wait at $q_0$ or to proceed to the accepting state $q_1$. The system can stay in $q_1$ as long as only $a$ is seen, and seeing a $!a$ forces the automaton (deterministically) to the rejecting sink $q_2$. For $\pi^\mathcal{S} = \pi$, satisfying  the specification $\F\G\, a$ on the SDS with probability $p>0$, if $\pi^{\Aut}$ always chooses to stay at $q_0$ after seeing every $a$, then the left side of the inequality will be $p$ but the right side will be $0$.

\smallskip\noindent{\bf Rejecting states.} An LDBA $\Aut = \tup{\Q,\qinit,\Sigma,\tran,\acc}$ induces a graph over its states with edges defined by all possible nondeterministic transitions in the automaton, i.e., it induces a graph $G = (\Q, E)$ with $E = \{(q,q') \in \Q\times\Q \mid \exists \sigma \in \Sigma.\, q' \in \Delta(q,\sigma)\}$. 

\begin{wrapfigure}{r}{0.5\linewidth}
	\centering
	\begin{tikzpicture}
		\node[state,initial,align=center]	(q0)		at	(0,0)	{$q_0$};
		\node[state,accepting,align=center]	(q1)	[right=of q0]	{$q_1$};
		
		\path[->]
			(q0)	edge[bend left]		node[above]	{$x\leq 0$}	(q1)
					edge[loop above]	node[above]	{$x > 0$}	()
			(q1)	edge[bend left]		node[below]	{$x>0$}		(q0)
					edge[loop above]	node[above]	{$x\leq 0$}	();
	\end{tikzpicture}
	\caption{The specifications $\G\F\,(x\leq 0)$ from Example~\ref{ex:random walk:1} represented using LDBA, which in this case is a deterministic B\"uchi automaton, where $q_1$ is the B\"uchi accepting state.}
	\label{fig:random walk: specs}
\end{wrapfigure}
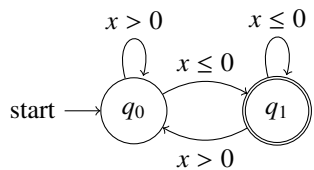
We denote by $\Qreject \subseteq \Q$ the set of all states in the LDBA from which no state in the set of accepting states $\acc$ can be reached in the graph $G$. Thus, if an infinite run $q_0,q_1,\dots$ in the LDBA contains a state in $\Qreject$, then it cannot be an accepting run as it will contain no states in $\acc$ after the first occurence of a state in $\Qreject$. To that end, we call $\Qreject$ the set of {\em rejecting states} in the LDBA.

\begin{example}\label{ex:random walk:2}
	The specification $\G\F\,(x\leq 0)$ from the 1D random walk is shown in the LDBA form in Figure~\ref{fig:random walk: specs}.
	The atomic predicates are given by $(x\leq 0)$ and $(x>0)$.
	The product state remains in the automaton state $q_0$ until $x>0$ is true, and each time the target $x\leq 0$ is true, the automaton visits the state $q_1$, which is the accepting state.
\end{example}
\section{Supermartingale Certificate for $\omega$-regular Specifications}\label{sec:theory}

We now present the theory behind our novel supermartingale certificates for quantitative $\omega$-regular specifications, which is the main technical contribution of our work.
Definition~\ref{def:certificate} formally defines the certificates and Theorem~\ref{thm:soundness} establishes their soundness. 
The computational aspects of the certificate will be discussed in Section~\ref{sec:algo}. We also assume that a controller is provided as input, and in Section~\ref{sec:algo} we will show how to compute it when it is not provided.

In what follows, suppose the following are given: an SDS $\Sys=\tup{\SS,\IS,\DS,d,f,\Init}$, a controller $\pi$ in $\Sys$, an $\omega$-regular specification $\varphi$ over a set of atomic propositions $\AP$, and a minimum probability  $p \in [0,1]$ with which the specification should be satisfied. 

\subsection{Intuitive Overview}
\label{sec:theory:overview}

 Let $\Aut = (\Q,\qinit,2^P,\tran,\acc)$ be an LDBA that accepts the same set of infinite words over $2^P$ as the specification $\varphi$. Our certificate is defined over the product SDS $\Sysx$ of the SDS $\Sys$ and the LDBA $\Aut$. By Proposition~\ref{thm:product}, this reduces the problem of certifying
\[ \centering \textit{the specification $\varphi$ is satisfied with probability $\geq p$} \]
to the problem of certifying that 
\begin{equation*}
\begin{split}
	\centering &\textit{there exists a controller $\pi^\times = (\pi^\Sys,\pi^\Aut)$ in $\Sysx$ with $\pi^\Sys = \pi$, such that} \\
	 \textit{the set } &\textit{of states $\mathbb{R}^n \times \acc$ in $\Sysx$ is reached infinitely many times with probability $\geq p$.} \\
\end{split}
\end{equation*}
To prove this, it suffices to show that there exists a controller $\pi^\times = (\pi^\Sys,\pi^\Aut)$ in $\Sysx$ with $\pi^\Sys = \pi$, such that the following two properties are satisfied:
\begin{compactenum}
	\item {\em Safety property.} There exists a set of states $\SI$ in the product SDS $\Sysx$ that contains all initial states $\Init \times \{q_0\}$, no rejecting states $\mathbb{R}^n \times \Qreject$, and such that $\SI$ is left with probability at most $1-p$. In other words, the pair $(\SI,p)$ is a {\em stochastic invariant}~\cite{ChatterjeeNZ17}. 
	\item {\em Liveness property.} A random run in the product SDS $\Sysx$ either visits $\mathbb{R}^n \times \acc$ infinitely many times or eventually leaves the set $\SI$ with probability $1$.
\end{compactenum}
The conjunction of the safety and the liveness properties guarantees that the set of states $\mathbb{R}^n \times \acc$ is visited infinitely many times with probability at least $p$.

The safety-liveness decomposition of the problem is reflected in the design of our supermartingale certificate. 
In particular, our certificate $\Cert$ consists of two components $\CertS$ and $\CertB$, called the \emph{safety certificate} and the \emph{liveness certificate}, which are connected in a hierarchical fashion:
the certificate $\CertS$ is independent of $\CertB$ and generates the stochastic invariant $\SI$, whereas the certificate $\CertB$ depends on $\CertS$ to use the set $\SI$ as a ``safety net'' and make sure that the liveness condition is fulfilled.
A crucial aspect of this hierarchical composition is that the two certificates agree on their choices of control inputs and resolutions of the nondeterminisms in the LDBA at all time.
From this, we obtain a single joint controller $\pi^\times = (\pi^\Sys,\pi^\Aut)$ in $\Sysx$ with $\pi^\Sys = \pi$, and formally certify that both the safety and the liveness property specified above are satisfied.

\smallskip\noindent{\bf Invariants.}  Supermartingale certificates need to satisfy a set of conditions in the reachable parts of the state space, starting from the initial states. Since computing the exact set of reachable states is in general infeasible, we define our supermartingale certificates with respect to an over-approximation, called a {\em (supporting) invariant}. In this section, we assume that the invariant is given, and in Section~\ref{sec:algo}, we will describe how to automatically compute it along with the certificate and the controller. 

Formally, a state $x$ in the SDS $\Sys$ is {\em reachable}, if there exists a path $x_0,x_1,\dots$ in $\Sys$ with $x_t = x$ for some $t \in \mathbb{N}_0$. An {\em invariant} in $\Sys$ is a set of states $I \subseteq \SS$ which contains all reachable states in $\Sys$. The definition of an invariant in a product SDS $\Sysx$ is analogous.

Invariants are integral components of supermartingale certificates, because they determine within which part of the state space the certificate conditions are required to be fulfilled; outside of the invariant, the certificates may behave arbitrarily.
Therefore, our formal definition of \lsm certificates will include their supporting invariants as their ``domains'' of operation.

\subsection{Supermartingale Certificates for LDBA Specifications}

To formalize the certificate components $\CertS$, $\CertB$, and $I$, as described above, we build upon and generalize the existing supermartingale certificates for reasoning about quantitative safety and qualitative recurrence (i.e., B\"uchi) properties.

\smallskip\noindent{\bf Safety certificates.} Our safety certificate $\CertS$ draws insight from \emph{repulsing supermartingales}~\cite{ChatterjeeNZ17}, which were introduced for reasoning about quantitative safety. Intuitively, given some set of unsafe states $U \subseteq \SS$ in an SDS $\Sys$ under a controller $\pi$, a repulsing supermartingale for $U$ is a function $\CertS: \SS \rightarrow \mathbb{R}$, along with an invariant $I$, which to every system state assigns a real value that is required to satisfy the following conditions at the states in $I$: \textbf{(S1)}~at initial states, the value of $\CertS$ is below some negative threshold $\eta \leq 0$, \textbf{(S2)}~at unsafe states in $U\cap I$, the value of $\CertS$ is above the threshold $0$, and \textbf{(S3)}~at states in $I$ in which $\CertS \leq 0$, the value of $\CertS$ must strictly decrease in expectation by $\epsilon > 0$ upon every one-step execution of the SDS under the controller $\pi$, while also ensuring that the absolute change in its value lies within some interval of length $M > 0$, and ensuring that the next state remains within $I$. It was shown that, if a repulsing supermartingale for the unsafe set exists, then the SDS stays within the set of states $\SI = \{x \in \SS \mid \CertS(x) < 0\}$---the \emph{stochastic invariant}---and thus does not reach the unsafe set $U$ with probability at least $p \geq 1 - \exp(\frac{8 \cdot \epsilon \cdot \eta}{M^2})$~\cite{WangS0CG21} (the bound of~\cite{WangS0CG21} is an improved version of the bound in~\cite{ChatterjeeNZ17}).
	
For us, the set of unsafe states of the product SDS $\Sysx$ is given by $U = \mathbb{R}^n \times \Qreject$, where recall that $\Qreject$ is the set of rejecting states in the LDBA, and our goal is to ensure that there exists a controller $\pi^\times_\safe = (\pi^\Sys,\pi^\Aut)$ in $\Sysx$ with $\pi^\Sys = \pi$, under which the product SDS is safe with probability at least $p$.
To this end, we extend repulsing supermartingales by requiring condition (S3) to be satisfied only for {\em some} outgoing LDBA transition at a state in $\Sysx$---call it condition \textbf{(S3+)}.
The stochastic invariant also gets generalized to the state space of the product SDS in the obvious way: $\SI^\times = \set{(x,q)\in \SS\times \Q\mid \CertS(x,q) < 0}$.   The controller $\pi^\Aut$ then picks that LDBA transition towards satisfying the quantitative safety property.

\smallskip
\noindent{\bf Liveness certificates.} Our liveness certificate $\CertB$ draws insight from the supermartingale certificate of Chakarov et al.~\cite{ChakarovVS16} for proving that a given set of target states $T\subseteq \SS$ in $\Sys$ under a given controller $\pi$ is visited infinitely many times with probability~$1$. Intuitively, a supermartingale certificate for $T$ is a function $\CertB: \SS \rightarrow \mathbb{R}$, along with an invariant $I$, which to every system state assigns a real value that is required to satisfy the following conditions at the states in $I$: \textbf{(L1)}~the value of $\CertB$ is non-negative, \textbf{(L2)}~at non-target states outside $T$, the value of $\CertB$ must strictly decrease in expected value by $\epsilon > 0$ upon every one-step execution of the system under the controller $\pi$, and \textbf{(L3)}~at target states in $T \cap I$, the value of $\CertB$ is allowed to increase in expectation by at most $M>0$ in every one-step execution of the system under the controller $\pi$, while making sure that the next state does not go outside of $I$.
 If such a certificate for the target set $T$ exists, then $T$ is reached infinitely many times with probability $1$~\cite{ChakarovVS16}.

For us, the set of target states of the product SDS $\Sysx$ is given by $T=\mathbb{R}^n\times\acc$, where recall that $\acc$ is the set of B\"uchi accepting set in the LDBA.
In our \lsm certificate, we extend the certificate of Chakarov et al.~\cite{ChakarovVS16} in two ways. First, we present a supermartingale certificate for proving that, with probability~$1$, {\em either} the B\"uchi accepting set $\SS\times\acc$ of the product SDS is reached infinitely many times {\em or} the stochastic invariant set $\SI^\times$ is eventually left---call it condition \textbf{(L3+)}. 
This is in line with the overview in Section~\ref{sec:theory:overview} and our requirements for a liveness certificate. Second, in the product MDP $\Sysx$, our goal is again to ensure that there exists a controller $\pi^\times_\live = (\pi^\Sys,\pi^\Aut)$ in $\Sysx$ with $\pi^\Sys = \pi$, under which this property is satisfied. This is accounted for by requiring conditions (L2) and (L3) of the certificate of Chakarov et al.~\cite{ChakarovVS16} to be satisfied only for {\em some} outgoing LDBA transition at a state in $\Sysx$.

\smallskip
\noindent\textbf{The hierarchical composition of safety and liveness certificates.}
The challenge in combining the safety and liveness components lies in ensuring that the properties implied by the two components are satisfied with respect to the {\em same controller} $\pi^\times = (\pi^\Sys,\pi^\Aut)$ in $\Sysx$ with $\pi^\Sys = \pi$; in other words, we need to ensure that the controllers $\pi^\times_\safe$ and $\pi^\times_\live$ are the same. 
Otherwise, we cannot conclude that the set of states $\mathbb{R}^n \times \acc$ is reached infinitely many times with probability at least $p$. 
In particular, we need to ensure that condition (S3+) on one hand, and conditions (L2) and (L3+) on the other hand, are always satisfied with respect to the same resolution of non-determinism, i.e.,~with respect to the same transition $q' \in \Delta(q,a)$ in the LDBA. We will pinpoint how this is achieved after presenting our novel \lsm certificate in the following.

\begin{definition}[\lsm certificates]\label{def:certificate}
	Consider an SDS $\Sys$ and a controller $\pi$ in it. Let $\Aut = \tup{\Q,\qinit,2^P,\tran,\acc}$ be an LDBA, $\Sysx$ be the product SDS of $\Sys$ and $\Aut$, and $I\subseteq \SS\times \Q$ be an invariant in $\Sysx$.

	An {\em \lsm certificate} is a triple $\Cert = (\CertS, \CertB, I)$, where $\CertS: \SS \times \Q \rightarrow \mathbb{R}$ and $\CertB: \SS \times \Q \rightarrow \mathbb{R}$ assign real values to each state in $\Sysx$, and $I\subseteq \SS\times \Q$ is an invariant in $\Sys^\times$. We require that there exist $\eta^S \leq 0$, $\epsilon^S > 0$, $M^S>0$, and $\beta^S \in \mathbb{R}$, as well as $\epsilon^L>0$, $M^L > 0$, such that the following conditions are satisfied:
	\begin{enumerate}[(a)]
			\item {\em Initial condition of the invariant:} the initial state is inside the invariant, i.e.,
				\[ \forall (x,q)\in \Init^\times.\, (x,q)\in I. \] \label{constraint:cert:a-}
			\item {\em Initial condition of the safety certificate:} the value of $\CertS$ at initial states is at most $\eta^S$, i.e.,
				\[ \forall x \in \Init.\,   \CertS(x, \qinit) \leq \eta^S.\tag{\text{{\color{black!30!white}Cond. \textbf{S1}}}}\] \label{constraint:cert:a}
			\item {\em Safety condition of the safety certificate:} the value of $\CertS$ at reachable rejecting states is non-negative, i.e.,			\[ \forall (x,q) \in (\SS\times\Qreject)\cap I.\,  \CertS(x, q) \geq 0. \tag{\text{{\color{black!30!white}Cond. \textbf{S2}}}}\]\label{constraint:cert:b}
			\item {\em Non-negativity condition of the liveness certificate:} the value of $\CertB$ is non-negative at states in the invariant $I$, which over-approximates the set of reachable states, i.e.,
			\[ \forall (x,q)\in I.\,  \CertB(x, q) \geq 0. \tag{\text{{\color{black!30!white}Cond. \textbf{L1}}}}\] \label{constraint:cert:c}
			\item {\em Strict expected decrease of the liveness certificate until $\SS \times \acc$ or} $\SI = \{(x,q) \in \Sysx \mid \CertS(x,q) < 0\}$ \emph{is reached with the safety condition:} we have
			\begin{equation*}
				\begin{split}
					&\forall (x,q) \in \SS\times(\Q \setminus (\acc \cup \Qreject))\cap I.\, \CertS(x,q) \leq 0 \Longrightarrow \\
					&\bigvee_{a \in \Sigma}\bigvee_{q' \in \tran(q,a)}
					\left[	
					\begin{array}{c}
						\overbrace{\left(
							\begin{array}{c}
								\Big[ x\ \models a \land \forall w\in \DS.\, (f(x,\pi(x),w),q')\in I \land \\
								 \CertS(x,q) \geq \mathbb{E}^{\pi}_{w}[\CertS(f(x,\pi(x),w),q')] + \epsilon^S \,\land \\
								 \forall w \in \DS.\, \beta^S \leq \CertS(x,q) - \CertS(f(x,\pi(x),w),q') \leq \beta^S + M^S \Big]
							\end{array}
						\right)}^{\safetycond\quad \textup{{\color{black!30!white}Cond. \textbf{S3+}}}}\\
						\bigwedge		\\	
						\underbrace{\CertB(x,q) \geq \mathbb{E}^{\pi}_{w}[\CertB(f(x,\pi(x),w),q')] + \epsilon^L}_{\textup{{\color{black!30!white}Cond. \textbf{L2}}}} 				
					\end{array}
					\right]
				\end{split}
			\end{equation*}\label{constraint:cert:d}
			\item {\em Bounded expected increase of the liveness certificate in $(\SS \times \acc) \backslash \SI$ with the safety condition:} we have
			\begin{equation*}
				\begin{split}
					&\forall (x,q) \in (\SS\times \acc)\cap I.\,  \CertS(x,q) \leq 0 \Longrightarrow  \\
					& \bigvee_{a \in \Sigma, \tran(q,a)=q'} 
					\left[ 
						\begin{array}{c}
							\overbrace{\safetycond}^{\textup{{\color{black!30!white}Cond. \textbf{S3+}}}}\\
							\bigwedge\\
							\underbrace{\mathbb{E}^{\pi}_{w}[\CertB(f(x,\pi(x),w),q')] \leq \CertB(x,q) + M^L}_{\textup{{\color{black!30!white}Cond. \textbf{L3+}}}}\\							
						\end{array}											  			
					\right].
				\end{split}
			\end{equation*}\label{constraint:cert:e}
	\end{enumerate}
\end{definition}
Conditions~(b), (c), and (d) in Def.~\ref{def:certificate} formalize the (S1), (S2), and (L1) conditions that were informally introduced beforehand.
Conditions~(e) and (f) combine (S3+) with (L2) and (L3+), respectively, which are concerned about states that are, respectively, outside of $\acc\cup \Qreject$ and inside of $\acc$.
Additionally, in both (e) and (f), the left hand sides of the implications require that the state $(x,q)$ is in the stochastic invariant $\SI$, which is formalized by the inequality $\CertS(x,q)\leq 0$, and the right hand sides require the satisfaction of (S3+) which is captured using the formula $\safetycond$.
The formula $\safetycond$ encodes that the strict expected decrease and the bounded differences conditions are satisfied at a state $(x,q)$ of $\Sysx$, upon taking the LDBA transition to $q'$ by using $a$. 
The assertion $x \models a$ guarantees the satisfaction of $a$ at $x$, ensuring the availability of the LDBA transition to $q' \in \Delta(q,a)$. 
The clause $\CertS(x,q) \geq \mathbb{E}^{\pi}_{w}[\CertS(f(x,\pi(x),w),q')] + \epsilon^S$ encodes strict expected decrease, i.e.,~that the value of $\CertS$ decreases in expected value by at least $\epsilon^S$ upon one-step execution of the SDS under the controller $\pi$. 
The clause $\forall w \in \DS.\, \beta^S \leq \CertS(x,q) - \CertS(f(x,\pi(x),w),q') \leq \beta^S + M^S$ encodes the bounded differences condition, i.e.,~that for every stochastic disturbance $w \in \DS$, the change in the value of $\CertS$ is contained in the interval $[\beta^S,\beta^S+M^S]$.
Finally, the clause $\forall w\in \DS.\, (f(x,\pi(x),w),q')\in I$ guarantees that after the one-step execution of the SDS under the controller $\pi$, the new state is within the invariant $I$.
Moreover, the right hand sides of the implications of (d) and (e) also respectively implement (L2) and (L3+), which are straightforward.

The following theorem establishes the soundness of our certificate for proving satisfaction of quantitative $\omega$-regular specifications in SDSs.

\begin{theorem}[Soundness of \lsm certificates]\label{thm:soundness}
	Consider an SDS $\Sys$ and a controller $\pi$. Let $\varphi$ be an $\omega$-regular specification in $\Sys$, $p \in [0,1]$ be a probability threshold, and $\Aut = (\Q,\qinit,2^\AP,\tran,\acc)$ be an LDBA that accepts the same language as $\varphi$. Let $\Sysx$ be the product SDS of $\Sys$ and $\Aut$.
	
	Suppose that there exists an \lsm certificate $\Cert = (\CertS, \CertB, I)$ with the associated constants $\eta^S$, $\epsilon^S$, and $M^S$, such that $p \leq 1 - \exp{\frac{8\cdot\eta^S\cdot\epsilon^S}{(M^S)^2}}$. 
	Then, there exists a controller $\pi^\times = (\pi^\Sys,\pi^\Aut)$ in $\Sysx$ with $\pi^\Sys = \pi$, such that for every initial state $(x_0,\qinit) \in \Initx$, we have
	\begin{equation}\label{eq:soundness}
		\mathbb{P}^{\Sys,\pi}_{x_0}\Big[ \varphi \Big] \geq \mathbb{P}^{\Sysx,\pi^\times}_{(x_0,\qinit)}\Big[ \Buchi(\mathbb{R}^n \times \acc) \Big] \geq p.
	\end{equation}
\end{theorem}

\begin{proof}
	We first define $\pi^\Aut: \Sysx \rightarrow \Q$ in order to specify the controller $\pi^\times = (\pi^\Sys,\pi^\Aut)$ in $\Sysx$ with $\pi^\Sys = \pi$. In what follows, we fix an ordering of the states in $Q$. For each $(x,q) \in \Sysx$, we define $\pi^\Aut(x,q)$ as follows:
	\begin{compactitem}
		\item If $(x,q) \in I$, $\CertS(x,q) \leq 0$ and $q \in \Q \backslash (\acc \cup \Qreject)$, we define $\pi^\Aut(x,q)$ to be $q' \in \Q$ of the smallest order for which the predicate on the right-hand-side of condition~\eqref{constraint:cert:d} in Def.~\ref{def:certificate} is satisfied.
		\item If $(x,q) \in I$, $\CertS(x,q) \leq 0$ and $q \in \acc$, we define $\pi^\Aut(x,q)$ to be $q' \in \Q$ of the smallest order for which the predicate on the right-hand-side of condition~\eqref{constraint:cert:e} in Def.~\ref{def:certificate} is satisfied.
		\item If $(x,q) \in I$, $\CertS(x,q) \leq 0$ and $q \in \Qreject$, by condition (c) in Def.~\ref{def:certificate}, if $(x,q) \in I$ then $V^{\text{safe}}(x,q) \geq 0$. Hence, $V^{\text{safe}}(x,q)$ must be $0$. We define $\pi^\Aut(x,q)$ to be $q' \in \Q$ of the smallest order for which there exists a letter $a \in \Sigma$ with $x \models a$ and $q' \in \Delta(q,a)$. Such $q' \in \Q$ is guaranteed to exist, as it is always possible to take at least one transition in the LDBA.
		\item If $(x,q) \not\in I$ or $\CertS(x,q) > 0$, we define $\pi^\Aut(x,q)$ exactly as in the previous case above, i.e., $\pi^\Aut(x,q)$ is $q' \in \Q$ of the smallest order for which there exists a letter $a \in \Sigma$ with $x \models a$ and $q' \in \Delta(q,a)$.
	\end{compactitem}
	We note that fixing an order of states in $Q$ and picking a state $q'$ of the smallest order is done in order to ensure that $\pi^\Aut: \Sysx \rightarrow \Q$ is a measurable function, for the semantics of the controller to be mathematically well defined.
	
We now show that the controller $\pi^\times = (\pi^\Sys,\pi^\Aut)$ with $\pi^\Sys = \pi$ and $\pi^\Aut$ defined as above satisfies the theorem's claim. The first inequality follows from Proposition~\ref{thm:product}, hence we only need to prove the second inequality. We show that the certificate component $\CertS$ defines a repulsing supermartingale for the unsafe set $\mathbb{R}^n \times \Qreject$~\cite{ChatterjeeNZ17} in the product SDS $\Sysx$ controlled by our controller $\pi^\times = (\pi^\Sys,\pi^\Aut)$. To prove this, we need to show that $\CertS$ and $I$ satisfy conditions \textbf{(S1)}-\textbf{(S3)} defined earlier in this section. Condition~\textbf{(S1)} is implied by condition~\eqref{constraint:cert:a} in Definition~\ref{def:certificate}. Condition~\textbf{(S2)} is implied by condition~\eqref{constraint:cert:b} in Definition~\ref{def:certificate}. Finally, condition~\textbf{(S3)} is implied by conditions~\eqref{constraint:cert:d} and~\eqref{constraint:cert:e} in Definition~\ref{def:certificate}. Hence, $\CertS$ indeed defines a repulsing supermartingale for the unsafe set $\mathbb{R}^n \times \Qreject$, and $I$ is the respective invariant. Thus, by~\cite[Theorem~5.1]{WangS0CG21}, it follows that for every initial state $(x,\qinit) \in \SSx$, we have
\begin{equation}\label{eq:eq1}
	\mathbb{P}^{\Sysx,\pi^\times}_{(x_0,\qinit)} \Big[ \textit{Safe}\Big(\{(x,q) \in \SSx \mid \CertS(x,q) \geq 0\}\Big) \Big] \geq 1 - \exp{\frac{8\cdot\eta^S\cdot\epsilon^S}{(M^S)^2}} \geq p,
\end{equation} 
where the inequality indeed holds with respect to our controller $\pi^\times = (\pi^\Sys,\pi^\Aut)$ and $\textit{Safe}(U)$ denotes the set of all paths in $\Sysx$ that do not contain states from some set $U$.

Next, we show that, under our controller $\pi^\times = (\pi^\Sys,\pi^\Aut)$, either the set of states $\mathbb{R}^n \times \acc$ is visited infinitely many times or the set of states $\SI = \{(x,q) \in \SSx \mid \CertS(x,q) \leq 0\}$ is eventually left with probability $1$. To prove this, we first define a new product SDS $\Sysx_{\text{new}}$ by modifying its dynamics function $\Dynx\colon \SSx\times\ISx\times\DS\to \SSx$ to
\[ \Dynx_{\text{new}}((x,q),(u,q),w) = \begin{cases}
	\Dynx((x,q),(u,q),w), &\text{if } \CertS(x,q) < 0 \\
	(x,q), &\text{if } \CertS(x,q) \geq 0. 
\end{cases}
\]
Intuitively, we are redefining the product SDS to make the complement of $\SI$, i.e.,~the set of all states at which $\CertS(x,q) \geq 0$, into a sink where the system remains stuck forever once entered. Then, in order to prove that either the set of states $\mathbb{R}^n \times \acc$ is visited infinitely many times or the set of states $\SI = \{(x,q) \in \SSx \mid \CertS(x,q) \leq 0\}$ is eventually left with probability $1$, it suffices to prove that the set of states $(\mathbb{R}^n \times \acc) \cup \neg\SI$ in the {\em new product SDS} $\Sysx_{\text{new}}$ is visited infinitely many times with probability $1$. To show this, we prove that the certificate component $\CertB$ defines an instance of the supermartingale certificate of~\cite{ChakarovVS16} for proving that the target set $(\mathbb{R}^n \times \acc) \cup \neg\SI$ is reached infinitely many times with probability $1$ in the {\em new product SDS} $\Sysx_{\text{new}}$ controlled by our controller $\pi^\times = (\pi^\Sys,\pi^\Aut)$. To prove this, we need to show that $\CertB$ and $I$ satisfy conditions \textbf{(L1)}-\textbf{(L3)} defined earlier in this section. Condition~\textbf{(L1)} is implied by condition~\eqref{constraint:cert:c} in Definition~\ref{def:certificate}. Condition~\textbf{(L2)} is implied by condition~\eqref{constraint:cert:d} in Definition~\ref{def:certificate}. Finally, condition~\textbf{(L3)} is implied by conditions~\eqref{constraint:cert:e} in Definition~\ref{def:certificate} and by our construction of the new product SDS $\Sysx_{\text{new}}$, since it ensures that the value of $\CertB$ remains constant once a state in $\neg\SI$ is reached. 
Hence, $\CertB$ indeed defines an instance of the supermartingale certificate of~\cite{ChakarovVS16} for proving that the target set $(\mathbb{R}^n \times \acc) \cup \neg\SI$ is reached infinitely many times with probability $1$ in $\Sysx_{\text{new}}$, and it follows that in the product SDS $\Sysx$ we have
\begin{equation}\label{eq:eq2}
	\mathbb{P}^{\Sysx,\pi^\times}_{(x_0,\qinit)} \Big[ \textit{B\"uchi}(\mathbb{R}^n \times \acc) \cup \textit{Reach}(\neg \SI)  \Big] = 1,
\end{equation}
where the equality indeed holds with respect to our controller $\pi^\times = (\pi^\Sys,\pi^\Aut)$, $\textit{B\"uchi}(T)$ denotes the set of all paths in $\Sysx$ that visit states in $T$ infinitely many times, and $\textit{Reach}(T)$ denotes the set of all paths in $\Sysx$ that eventually visit a state in $T$.

Combining eq.~\eqref{eq:eq1} and eq.~\eqref{eq:eq2} above, we conclude that for every initial state $(x,\qinit) \in \SSx$ we have
\begin{equation*}
	\begin{split}
		\mathbb{P}^{\Sysx,\pi^\times}_{(x_0,\qinit)} \Big[ &\textit{B\"uchi}(\mathbb{R}^n \times \acc)\Big]  \geq \mathbb{P}^{\Sysx,\pi^\times}_{(x_0,\qinit)} \Big[ \textit{B\"uchi}(\mathbb{R}^n \times \acc) \cup \textit{Reach}(\neg \SI)  \Big]  - \mathbb{P}^{\Sysx,\pi^\times}_{(x_0,\qinit)} \Big[ \textit{Reach}(\neg \SI)  \Big] \\
		&\geq \mathbb{P}^{\Sysx,\pi^\times}_{(x_0,\qinit)} \Big[ \textit{B\"uchi}(\mathbb{R}^n \times \acc) \cup \textit{Reach}(\neg \SI)  \Big]  - \Big(1 -\mathbb{P}^{\Sysx,\pi^\times}_{(x_0,\qinit)} \Big[ \textit{Safe}(\neg \SI)  \Big]\Big) \\
		&\geq 1 - (1 - p) = p.
	\end{split}
\end{equation*}
This concludes the proof of the theorem claim.
\hfill\qed
	
\end{proof}

\begin{example}[\lsm certificates]
	Consider the 1D random walk from Example~\ref{ex:random walk:1} and \ref{ex:random walk:2} with specification $\G\F\,(x\leq 0)$.
	The following constitute an \lsm certificate:
	The invariant $I$ in $q_0$ is $73+(1/2)x$ and in $q_1$ is the constant $0$, $\CertS(\cdot,q_0)=\CertS(\cdot,q_1)=-9+(5/16)x$, $\CertB(\cdot,q_0) = (367/2) + (3/4)x$, $\CertB(\cdot,q_1) = (4747/128)-(1/256)x$, with the variables $\eta^S=-8$, $\epsilon^S=5/32$, $M^S=1$, $\beta^S = -11/32$, $\epsilon^L=3/8$, and $M^L = 260$.
	It can be checked that $\{\CertS,\CertB,I\}$ is a valid \lsm that satisfies all the constraints in Definition~\ref{def:certificate}.
\end{example}

\begin{remark}[On the incompleteness of \lsm certificates]\label{rem:incompleteness}
Our \lsm certificates are incomplete: it is possible that a certificate does not exist even though there exists a controller that fulfills the given $\omega$-regular specification.
The incompleteness of LDBSM certificates follows from the incompleteness of our safety (based on repulsing supermartingales~\cite{ChatterjeeNZ17}) and liveness certificates (based on almost-sure recurrence certificates of Chakarov et al.~\cite{ChakarovVS16}). The incompleteness of repulsing supermartingales was discussed in Chatterjee et al.~\cite{chatterjee2022sound}, and we omit the details. An example of incompleteness of the liveness certificate is when the target set is closed under system dynamics and cannot be left once entered. In this case, infinitely many visits to the target is equivalent to reaching it once. For almost-sure reachability, the liveness certificate satisfying conditions (L1) and (L2) (see Section~\ref{sec:theory:overview}) becomes a ranking supermartingale of Chakarov and Sankaranarayanan~\cite{chakarov2013probabilistic}. Ranking supermartingales are sound and complete for finite expected time reachability~\cite{fu2019termination}, but are incomplete for almost-sure reachability.
\end{remark}
\section{Polynomial Template-Based Synthesis Algorithms}\label{sec:algo}

We now present our algorithms for automated verification and control of stochastic dynamical systems via our \lsm certificates. While our theory on \lsm certificates in Section~\ref{sec:theory} is applicable to general SDS and controllers, the algorithms in this section are restricted to systems that can be specified via polynomial arithmetic.

\smallskip\noindent{\bf Assumption: Polynomial systems.} We assume that the system dynamics function $f$,  the controller $\pi$, and arithmetic expressions for all atomic propositions appearing in the $\omega$-regular specification $\varphi$ are {\em polynomial functions} over the state, control input, and stochastic disturbance variables. More generally, we allow the system dynamics function $f$ to be piecewise-polynomial. That is, it can be of the form
\begin{equation}\label{eq:piecewisepolynomial}
	f(x,u,w) = 
	\begin{cases}
		f_1(x,u,w), &\text{if } \land_{i=1}^{N_1} g_i^1(x) \leq 0, \\
		\,\,\,\,\,\,\dots & \\
		f_k(x,u,w), &\text{if } \land_{i=1}^{N_k} g_i^k(x) \leq 0, \\
	\end{cases}
\end{equation}
where all $f_i$'s and $g_i^j$'s are polynomial functions over their input variables.  Later, we will fix $d$ as the specified maximum degree of our polynomial templates, and we assume that the first $d$ moments of $f(x,u,W)$---with $W$ being the random variable representing the distribution over noise---be representible and given as polynomials over $x$ and $u$. For example, the system dynamics function in Example~1 is piecewise linear with $k=2$.

\smallskip\noindent{\bf Input.} Our verification and control algorithms take as input an SDS $\Sys$ whose system dynamics function is piecewise polynomial, an LDBA $\Aut$ for an $\omega$-regular specification $\varphi$ defined over a set of polynomial atomic prepositions $P$, and a minimum probability threshold $p \in [0,1]$ with which the $\omega$-regular specification $\varphi$ should be satisfied. The verification algorithm also takes as input the polynomial controller $\pi$.

In addition, the algorithms take as input two parameters: (1)~the polynomial degree $d$ of polynomial templates used in the synthesis of the \lsm certificate and the controller, and (2)~the number of polynomial inequalities $n_I$ used to define the supporting invariant. These two parameters are formally defined in Step~1 below.

The verification algorithm may accept polynomial supporting invariants as inputs, and if an invariant $I$ is provided, then the algorithm only computes $\CertS$ and $\CertB$ for the fixed $I$.
For the control algorithm, the invariant depends on the controller, and is therefore always synthesized as a part of the certificate.
For the uniformity of presentation, we will assume that the verification algorithm is not provided an input invariant.

\smallskip\noindent{\bf Output.} Both the verification and the control algorithms return as output the \lsm certificate $\Cert = (\CertS, \CertB, I)$, as defined in Definition~\ref{def:certificate}. 
The control algorithm also returns the SDS controller $\pi$ computed by the algorithm. If the algorithms fail due to, e.g.,~the SMT solver timing out or returning ``UNSAT,'' they return ``Unknown.''

\smallskip\noindent{\bf Algorithm outline.} Our algorithms employ the classical {\em template-based synthesis} approach to synthesize an \lsm certificate $\Cert = (\CertS, \CertB,I)$ (and also the controller $\pi$ for the control problem). This automation approach is similar to those employed in prior works on algorithmic synthesis of supermartingales for analyzing stochastic systems and programs with respect to reachability or safety properties~\cite{CFG16,ChatterjeeNZ17,takisaka2021ranking,chatterjee2022sound}.

In Step~1, both the verification and the control algorithms fix symbolic polynomial templates with \emph{unknown coefficients} for all objects that they need to compute. Namely, the verification algorithm fixes templates for the \lsm certificate $\Cert = (\CertS, \CertB, I)$, while the control algorithm also fixes a template for the controller $\pi$. The templates for $\CertS,\CertB$, and $\pi$ are polynomial \emph{expressions}, whereas the template for $I$ is a conjunction of $n_I$ polynomial inequalities of the form ``$P_I^j(x,q)\geq 0$,'' with each $P_I^j$ being a polynomial expression. 
In Step~2, the algorithms encode all certificate conditions in Definition~\ref{def:certificate} as a system of quantified polynomial entailments over the unknown template coefficients. In Step~3, the resulting system of quantified polynomial entailments is solved by employing existing algorithms and efficient tooling support provided by the \textsc{PolyQEnt} tool~\cite{chatterjee2024polyhorn}, which reduces the problem to solving a purely existentially quantified system of polynomial constraints for which an off-the-shelf SMT solver is used. In the rest of this section, we provide a more detailed description of these steps.

\smallskip\noindent{\bf Step~1: Setting up templates.} The algorithms fix symbolic polynomial templates for all objects that they need to compute:
\begin{compactitem}
	\item {\em \lsm certificate.}
	Suppose $x_1,\ldots,x_n$ represent the state variables for the $n$ dimensions.
	For each LDBA state $q\in \Q$, we set up the polynomial templates for $\CertS(\cdot,q)$ and $\CertB(\cdot,q)$, as well as for polynomial expressions $P_I^1(\cdot,q), \dots, P_I^{n_I}(\cdot,q)$ over the state variables, each of polynomial degree $d$ and with unknown coefficients whose values are to be found out in the subsequent steps of the algorithm. Here, the polynomial degree $d$ and the number $n_I$ of polynomial inequalities that define the invariant $I$ are algorithm parameters. 
	For example, if $n=2$, $d=2$, and $n_I=1$, then $\CertS(x_1,x_2,q) = \theta^{q,\safe}_{11}x_1^2 + \theta^{q,\safe}_{22}x_2^2 + \theta^{q,\safe}_{12}x_1x_2 + \theta^{q,\safe}_1x_1 + \theta^{q,\safe}_2x_2 + \theta^{q,\safe}_0$, $\CertB(x_1,x_2,q) = \theta^{q,\live}_{11}x_1^2 + \theta^{q,\live}_{22}x_2^2 + \theta^{q,\live}_{12}x_1x_2 + \theta^{q,\live}_1x_1 + \theta^{q,\live}_2x_2 + \theta^{q,\live}_0$, and $(x_1,x_2,q)\in I$ is expressed as $P_I^1(x,q)\geq 0$ where $P_I^1(x,q) = \theta^{q,I}_{11}x_1^2 + \theta^{q,I}_{22}x_2^2 + \theta^{q,I}_{12}x_1x_2 + \theta^{q,I}_1x_1 + \theta^{q,I}_2x_2 + \theta^{q,I}_0 \geq 0$.	
	\item {\em Certificate parameters.} 
	The algorithm introduces variables $\var{\eta^S}$, $\var{\epsilon^S}$, $\var{M^S}$, $\var{\beta^S}$, $\var{\epsilon^L}$, $\var{M^L}$,  whose values will represent the constant parameters $\eta^S \leq 0$, $\epsilon^S,M^S,\epsilon^L,M^L > 0$, and $\beta^S\in \mathbb{R}$ from Definition~\ref{def:certificate}.	
	\item {\em Controller (for the control algorithm).}
	For every LDBA state $q\in \Q$, we set up a polynomial controller $\pi(\cdot,q)$ over the state variables with degree $d$ and with unknown coefficients.
	For example, if $n=2$ and $d=2$, then $\pi(x_1,x_2,q) = \lambda^q_{11}x_1^2 + \lambda^q_{22}x_2^2 + \lambda^q_{12}x_1x_2 + \lambda^q_{1}x_1 + \lambda^q_2x_2 + \lambda^q_0 $, where $\lambda_{\cdot}$-s are the unknown coefficients.
\end{compactitem} 

We will use $\Omega$ to denote the set of all unknown coefficients (i.e., $\theta^{q,\safe}_{11},\theta^{q,\live}_{11},\theta^{q,I}_{11},\ldots$) in the polynomial templates for $\CertS$, $\CertB$, $I$, and $\pi$, and the variables (i.e., $\var{\epsilon^S}$, $\var{\beta^S}$, $\ldots$).
For a given set of unknowns such as $\Omega$, we will write $\llbracket \Omega\rrbracket$ to denote the set of all possible valuations of the unknowns in $\Omega$.

\smallskip\noindent{\bf Step~2: Constraint collection.}
The goal of this step is to symbolically evaluate the conditions in Definition~\ref{def:certificate}, and rearrange them in a way that we finally obtain a system of polynomial inequality predicates in the prenex normal form 
\begin{align}\label{eq:b}
	\exists\omega\in\llbracket \Omega\rrbracket\;.\;\bigwedge_i \underbrace{\left( \forall q\in\Q\;.\;\forall x\in \mathbb{R}^n\;.\;\forall w\in \mathbb{R}^p\;.\; \Phi_i^\omega(q,x,w) \implies \Psi_i^\omega(q,x,w)\right)}_{R_i},
\end{align}
where $\Phi_i^\omega$ and $\Psi_i^\omega$ are boolean combinations of polynomial inequalities over the variables $q$, $x$, and $w$ and unknown coefficients whose values conform with $\omega$. 
First, observe that the constraints~(a)-(d) of Definition~\ref{def:certificate} generate constraints that are trivially of the form $R_i$. 
We explain the derivation for constraint~\eqref{constraint:cert:d}, the procedure for~\eqref{constraint:cert:e} is analogous:
\begin{enumerate}[(A)]
	\item \textbf{Prenex normal form:} First, we need to take out the two quantifications ``$\forall w\in \DS$'' from within ``$\safetycond$.''
	We show the process for the clause ``$\forall w \in \DS.\, \beta^S \leq \CertS(x,q) - \CertS(f(x,\pi(x),w),q') \leq \beta^S + M^S$,'' call it $\gamma_1$, and for ``$\forall w\in \DS.\, (f(x,\pi(x),w),q')\in I$,'' call it $\gamma_2$, the procedure is identical.
	There are two possibilities:
	\begin{enumerate}[(i)]
		\item Suppose the safety component of the \lsm certificate $\CertS$ and $q'$'s invariant polynomials $\set{P_{I}^i(\cdot,q')}_{i\in [1;n_I]}$  are linear, i.e., $d=1$, the disturbance in the system is additive, i.e., $\forall x,u,w\;.\;f(x,u,w) = g(x,u)+w$ for some function $g$, and the domain $\DS$ of the disturbance is bounded within a given range $[w_{\min},w_{\max}]$. 
		Then, for $\gamma_1$, we can altogether suppress ``$\forall w\in \DS$'' and replace $\safetycond$ with
		\begin{align*}
			\left(
					\begin{array}{c}
					 x\ \models a \land \forall w\in \DS.\, (f(x,\pi(x),w),q')\in I \\
					 \land \CertS(x,q) \geq \mathbb{E}^{\pi}_{w}[\CertS(f(x,\pi(x),w),q')] + \epsilon^S \\
					 \land \left( \beta^S \leq \CertS(x,q) - \CertS(g(x,\pi(x)),q') - \CertS(w_{\min},q') \leq \beta^S + M^S\right) \\
					 \land \left( \beta^S \leq \CertS(x,q) - \CertS(g(x,\pi(x)),q') - \CertS(w_{\max},q') \leq \beta^S + M^S\right) 
					 \end{array}
					 \right).
		\end{align*}
		This transformation is sound since for every $w\in\DS$, by the additivity of the noise and by the linearity of $\CertS$, we have $\CertS(f(x,u,w)) = \CertS(g(x,u)+w) = \CertS(g(x,u)) + \CertS(w)$, which is bounded between $\CertS(g(x,u)) + \CertS(w_{\min})$ and $\CertS(g(x,u)) + \CertS(w_{\max})$ from the boundedness of the noise and the linearity of $\CertS$.
		We can use the same idea for the clause $\gamma_2$.
		\item In general, the above restrictions could be insufficient, in which case we resort to the following stricter requirement at the expense of added incompleteness of the solution.
		The idea is that, instead of allowing $w$ to depend on the choice of $a\in \Sigma$ and $q'\in \tran(q,a)$, for $\gamma_1$, we require the following the hold true:
		\begin{align*}
		\forall a\in\Sigma\;.\; &\forall q'\in \tran(q,a)\;.\;\forall w\in \DS\;.\; \\
		&\left( \beta^S \leq \CertS(x,q) - \CertS(f(x,\pi(x),w),q') \leq \beta^S + M^S\right).
		\end{align*}
		Let the expression above be written as $D(x,q)$ in short, and it is easy to see that $D(x,q)$ can be expressed in the desired form \eqref{eq:b}.
		We can now write the constraint~\eqref{constraint:cert:d} of Definition~\ref{def:certificate} as:
		\begin{align*}
					&D(x,q)\,\land\\
					&\forall x \in \mathbb{R}^n.\, \CertS(x,q) \leq 0 \Longrightarrow \bigvee_{a \in \Sigma}\bigvee_{q' \in \tran(q,a)}\\
					&\quad\left(\begin{array}{c}
					x\ \models a \land \forall w\in \DS.\, (f(x,\pi(x),w),q')\in I \\
					 \land\CertS(x,q) \geq \mathbb{E}^{\pi}_{w}[\CertS(f(x,\pi(x),w),q')] + \epsilon^S \\
					\land\,\CertB(x,q) \geq \mathbb{E}^{\pi}_{w}[\CertB(f(x,\pi(x),w),q')] + \epsilon^L
					\end{array}\right)
		\end{align*}
		The same idea will work for $\gamma_2$.
	\end{enumerate}
	\item \textbf{Simplifications.} After moving the quantification ``$\forall w\in \DS$'' from the inside of ``$\safetycond$'' to the outermost part of 2(b) and 2(c), we substitute all the variables and functions with the templates that we set up in Step 1.
	For the expected value $\mathbb{E}^{\pi}_{w}[\CertS(f(x,\pi(x),w),q')]$, we use the linearity of the expectation operator to break it down into a polynomial over $x$ whose coefficients are functions of the moments of the distribution over $w$ (assumed to be given as input).
	For instance, since $\CertS(x,q)$ is a polynomial over $x$, and $f(x,\pi(x),w) = C_2(x)w^2 + C_1(x)w + C_0(x)$ where $C_2, C_1, C_0$ are polynomials over $x$, it follows that $\CertS(f(x,\pi(x),w),q) = \hat{C}_2(x)w^2 + \hat{C}_1(x)w + \hat{C}_0(x)$ where $\hat{C}_2, \hat{C}_1, \hat{C}_0$ are polynomials over $x$.
	This gives us: $\mathbb{E}^{\pi}_{w}[\CertS(f(x,\pi(x),w),q)] = \hat{C}_2(x)\cdot\mathbb{E}^{\pi}_{w}[w^2] + \hat{C}_1(x)\cdot\mathbb{E}^{\pi}_{w}[w] + \hat{C}_0(x)$, where $\mathbb{E}^{\pi}_{w}[w^i]$ is the $i$-th moment of the distribution (therefore, a constant) of the noise.
\end{enumerate}

\smallskip\noindent{\bf Step~3: Constraint solving.} Finally, the system of quantified polynomial entailments collected in Step~2 is solved. This is achieved by first employing Putinar's theorem~\cite{putinar1993positive} (or Farkas' lemma~\cite{farkas1902theorie}, if all involved polynomials have degree~$1$) to translate each quantified polynomial entailment into a system of purely existentially quantified polynomial constraints over the unknown template variables, as well as fresh variables introduced by the Putinar's theorem translation. Since this step is standard in the template-based synthesis literature, including that on the synthesis of supermartingale ceritificates \cite{CFG16,ChatterjeeNZ17,takisaka2021ranking,chatterjee2022sound}, we omit the details. The resulting system of purely existentially quantified polynomial constraints over real-valued variables is then solved by using an SMT solver. In our implementation, the application of Putinar's theorem and reduction to SMT solving are achieved via the \textsc{PolyQEnt} tool~\cite{chatterjee2024polyhorn}. See~\cite{chatterjee2024polyhorn} for details, including on how boolean combinations of polynomial inequalities are handled prior to applying Putinar's theorem. 
It is straightforward to note that if the encoding \eqref{eq:b} is satisfiable, then we obtain an \lsm certificate, and in case of control problem, we also get a controller. Then from the soundness of the \textsc{PolyQEnt} tool~\cite{chatterjee2024polyhorn} that we use to solve \eqref{eq:b}, Theorem~\ref{thm:alg:soundness} follows.

\begin{theorem}[Algorithm soundness]\label{thm:alg:soundness}
	Suppose that the verification algorithm returns an \lsm certificate $\Cert = (\CertS, \CertB,I)$. Then, $\Cert$ is an \lsm certificate with respect to the invariant $I$ as in Definition~\ref{def:certificate}, and the SDS $\Sys$ under controller $\pi$ satisfies specification $\varphi$ with probability at least $p$.
	
	Suppose that the control algorithm returns an \lsm certificate $\Cert = (\CertS, \CertB,I)$ and an SDS controller $\pi$. Then, $\Cert$ is an \lsm certificate with respect to the invariant $I$ as in Definition~\ref{def:certificate}, and the SDS $\Sys$ under controller $\pi$ satisfies specification $\varphi$ with probability at least $p$.
\end{theorem}
Similar to existing template-based synthesis methods for polynomial program verification and controller synthesis~\cite{CFG16,AsadiC0GM21}, the runtime complexity of our vertification and control algorithms is in PSPACE when the polynomial degree $d$ is treated as a constant parameter. This is because the size of the final SMT query produced by our method is $\mathcal{O}(n_I* |Q|* n^d)$. As the final SMT query is a sentence in the existential theory of the reals, the parameterized PSPACE runtime complexity follows.
\section{Experiments}\label{sec:experiments}

We implemented our algorithm as a proof-of-concept tool in Python, where our tool automatically translates the given inputs to quantified polynomial entailments by using the procedure outlined in Section~\ref{sec:algo}, and then uses the tool \textsc{PolyQEnt}~\cite{chatterjee2024polyhorn} as its back end to compute solutions.
\textsc{PolyQEnt} uses Putinar's theorem~\cite{putinar1993positive} (or Farkas' lemma~\cite{farkas1902theorie} for linear systems and certificate and controller templates) and translates the provided entailments into existentially quantified polynomial constraints over unknown template variables, and uses off-the-shelf SMT solvers Z3~\cite{de2008z3} and MathSAT~\cite{bruttomesso2008mathsat} to find satisfying assignments of the variables. Our prototype tool is availabe at \url{https://github.com/Ipouyall/Omega-Regular-Stoch-Cert}.

We used our tool to automatically compute \lsm certificates for the verification of the random walk SDS $\Sysrw$ from Example~\ref{ex:random walk:1}, and for the control of a modified version of $\Sysrw$ defined as follows.
The system $\Sysrw_c$ has the state space $\mathbb{R}$, input space $\IS=[-2,2]$, disturbance space $\DS=[0,1]$ with a continuous uniform distribution over $\DS$, system dynamics given as 
\begin{align}
	f_c(x,u,w) = 
		\begin{cases}
			x	&	\text{if } x>100,\\
			x + u + w	&	\text{otherwise,}
		\end{cases}
\end{align}	
and the set of initial states $[2,3]$.
Both the verification and the control example were solved with a number of different specifications as listed in Table~\ref{tab:experiments}. For each specification, the general requirement is to push the state towards the origin (eventually for $\F\, a$, always eventually for $\G\F\, a$, etc.).
For the verification problems, the disturbances in $W=[-2,1]$ create an automatic bias towards the origin, which ultimately helps satisfying the specifications.
For the control problems, the disturbances in $W=[0,1]$ create a challenge by introducing bias \emph{away} from the origin, so the controller needs to counteract with negative control inputs at each step.

All the experiments were performed on an Apple Macbook Air with M2 chip and 16 GB RAM, and
we report the results in Table~\ref{tab:experiments}.
In both the verification and control problems, the SDS's have small probabilities of violating the objectives, by reaching the region $x>100$ and not being able to come out afterwards.
Therefore, existing approaches supporting only qualitative (almost sure satisfaction) $\omega$-regular specifications would not work, and no single supermartingale-based approach for quantitative specifications would support all the specifications at once.
In contrast, our tool achieved $99\%$ probability for satisfying the specifications, and computed \lsm certificates (and controllersfor the control problem) within minutes.

\begin{table}
\caption{Experimental results for verification and control of the random walk examples.
The predicates are: $a \coloneqq x\leq 0$ and $b\coloneqq x\leq 100$. 
Each experiment was run with polynomial degree $d=1$.}
\label{tab:experiments}
\centering
	\begin{tabular}{wc{4cm} wc{2cm} wc{2cm} wr{2cm}}
		\toprule
			&	Specification	&	Probability	&	Time (s)\\
		\midrule
		\multirow{6}{*}{Verification (SDS $\Sysrw$)} 	&	$\F\, a$	&	$0.9999$	&	$0.35$	\\
		&	$\G\F\, a$	&	$0.9999$	& $27.28$	\\
		&	$b\U\, a$	&	$0.9999$	&	$6.13$ \\
		&	$\G b\land \F a$	&	$0.9999$	&	$27.78$	\\
		&	$\G b$		&	$0.9999$	&	$28.60$	\\
		&	$\F a\land \F b$	&	$0.9999$	&	$172.28$	\\
		\midrule
		\multirow{5}{*}{Control (SDS $\Sysrw_c$)} 	&	$\F\, a$	&	$0.9999$	&	$0.39$	\\
		&	$\G\F\, a$	&	$0.9999$	&	$4.90$	\\
		&	$b\U\, a$	&	$0.9999$	&	$18.90$	\\
		&	$\G b\land \F a$	&	$0.9999$	&	$46.68$	\\
		&	$\G b$		&	$0.9999$	&	$4.04$	\\
		\bottomrule
	\end{tabular}
\end{table}
\section{Related Work}\label{sec:relatedwork}

Earliest known types of certificates in control theory date back to mid-twentieth century, and were developed for certifying stability~\cite{lasalle1961stability,kalman1960control} and invariance~\cite{nagumo1942lage} of dynamical systems.
However, back then, obtaining suitable certificates for general nonlinear systems was largely a manual process.
It was in early 2000s, when the seminal work of Parrilo~\cite{parrilo2000structured} first proposed the use of Positivstellens\"atze to automate the process of certificate computation, which were later strengthened and refined~\cite{papachristodoulou2005tutorial,sloth2016computation}.

In parallel to control theory, from early 2000s, program analysis also experienced an emergence of certificate-based approaches for proving termination of deterministic programs~\cite{coloon2001synthesis,colon2002practical}.
In the landmark paper by Chakarov~and~Sankarnarayanan~\cite{chakarov2013probabilistic}, ranking functions for deterministic programs were generalized to ranking supermartinagles for probabilistic programs. This has sparked interest in using supermartingale certificates for the analysis of probabilistic programs, with signifiant advances in supermartingale certificates and automated algorithms for termination~\cite{CFNH16:prob-termination,CFG16,mciver2017new,abate2021learning,chatterjee2023lexicographic,majumdar2025sound}, safety~\cite{CS14:expectation-invariants,ChatterjeeNZ17,takisaka2021ranking,WangS0CG21,chatterjee2022sound}, B\"uchi and co-B\"uchi (stability)~\cite{ChakarovVS16}, and cost analysis~\cite{ngo2018bounded,WFGCQS19,chatterjee2024quantitative}. These results were also exported to control theory to solve the verification and control problem for stochastic dynamical systems with logical specifications~\cite{prajna2004stochastic,xue2023reach,vzikelic2023learning}.

A majority of existing supermartingale certificates for stochastic systems only support fragments of $\omega$-regular specifications. Only recently, Abate et al.~\cite{abate2024stochastic} proposed a new class of supermartingale certificates for $\omega$-regular specifications represented using Streett automata. However, their certificate is restricted to the probability $1$ satisfaction of the specification, contrary to our quantitative certificate that allows arbitrary probability thresholds to be imposed as a requirement for satisfying the specification.
It is important to note that our certificate is incomplete, like most other existing supermartingale certificates that are available in the literature, see Remark~\ref{rem:incompleteness}.

It is important to mention that there are also other non-certificate-based approaches to verification and control of stochastic systems.
For example, there is a long line of research that falls under the abstraction-based control paradigm, where a given system model is abstracted to a simpler graph, so that graph theoretic techniques can be applied to approximately solve the original verification and control problem~\cite{esmaeil2013adaptive,zamani2014symbolic,mallik2017compositional,dutreix2022abstraction,majumdar2024symbolic}.
Among these works, only one is known to be able to support quantitative $\omega$-regular specifications~\cite{dutreix2022abstraction}, though their functionalities are complementary to ours: 
They support general nonlinear systems but only having additive control input and additive stochastic noise with unimodal, symmetric distribution.
We support polynomial systems and controllers and stochastic noise that can have arbitrary distributions.

\section{Discussions}

We present \lsm certificates, the first supermartingale-based certificates for quantiative $\omega$-regular verification and control of infinite-state stochastic systems.
\lsm certificates are defined on the product of the system and the equivalent LDBA representation of the given $\omega$-regular specification, and they generalize and combine existing safety and liveness supermartingale certificates in a hierarchical manner to achieve soundness.
We supplement the theoretical development with a standard template-based synthesis algorithm for solving the verification and control problems in polynomial stochastic systems.

There are several interesting future directions.
First, we will explore alternate methods for computing \lsm certificates, including learning-assisted approaches~\cite{chatterjee2023learner,vzikelic2023learning}, alongside optimizing our template-based approach that relies on off-the-shelf constraint solvers.
Even the state-of-the-art constraint solvers struggled to compute \lsm certificates, and eliminating this tooling bottleneck will be a priority in future works.
Second, in the same vein as the previous one, we will explore compositional approaches for \lsm certificates for better scalability; such approaches are well-studied for existing classes of certificates~\cite{nejati2022compositional,vzikelic2024compositional}.
Finally, we will explore completeness questions.
Right now, \lsm certificates provide sound but incomplete solutions: if they exist, then it follows that the specification holds, but not the other way round.

\begin{credits}
\subsubsection{\ackname} This work was supported in part by the Singapore Ministry of Education (MOE) Academic Research Fund (AcRF) Tier 1 grant (Project ID:22-SIS-SMU-100) and the ERC project ERC-2020-AdG 101020093.
\end{credits}

\bibliographystyle{splncs04}
\bibliography{ref}


\end{document}